\documentclass[a4paper,fleqn]{cas-dc}

\usepackage[numbers,sort&compress]{natbib}
\usepackage{amsmath,amssymb,amsfonts,amsthm,mathtools}
\usepackage{graphicx}
\usepackage{epsfig}
\usepackage{array}
\usepackage{multirow}
\usepackage{dsfont}
\usepackage{xcolor}
\usepackage{soul}
\usepackage{url}
\usepackage[ruled]{algorithm2e}
\usepackage{algorithmic}
\usepackage{subfig}
\usepackage{etoolbox}
\allowdisplaybreaks

\graphicspath{{images/}{figs/}}

\newtheorem{theorem}{Theorem}
\newtheorem{lemma}{Lemma}

\newtheorem{proposition}{Proposition}
\newtheorem{corollary}{Corollary}[theorem]

\newtheorem*{remark*}{Remark}

\newcommand{\B}{}
\newcommand{\transport}[2]{u_{#1#2}}
\newcommand{\transportCapacity}{U}
\newcommand{\resourceCapacity}{N}

\DeclareMathOperator{\sgn}{sgn}

\begin{document}
\let\WriteBookmarks\relax
\def\floatpagepagefraction{1}
\def\textpagefraction{.001}

\shorttitle{Storage and Transport Capacity Design for a Self-Reliable Two-Node Stochastic Resource System}
\shortauthors{Dey et~al.}

\title[mode=title]{Storage and Transport Capacity Design for a Self-Reliable Two-Node Stochastic Resource System}\tnotemark[1]

\tnotetext[1]{This work is supported by Advanced Research Projects Agency-Energy OPEN through the project titled ``Rapidly Viable Sustained Grid'' via grant no. DE-AR0001016 and the US Department of Energy project titled ``CyDERMS: Center for Cybersecurity \& Resiliency of Distribution Energy Resources and Microgrids-integrated Distribution Systems'' via award no. DE-CR0000040.}

\author[1]{Arnab Dey}\fnmark[1]
\ead{dey00011@umn.edu}

\author[2]{Vivek Khatana}\cormark[1]\fnmark[1]
\ead{vkhatana@illinois.edu}

\author[3,4]{Ankur Mani}
\ead{ankur.mani@gmail.com}

\author[1]{Murti V. Salapaka}
\ead{murtis@umn.edu}

\fntext[fn1]{Equal contribution.}

\affiliation[1]{organization={Department of Electrical and Computer Engineering, University of Minnesota},
                city={Minneapolis},
                state={MN},
                country={USA}}

\affiliation[2]{organization={Department of Mechanical Science and Engineering, University of Illinois Urbana-Champaign},
                city={Urbana},
                state={IL},
                country={USA}}

\affiliation[3]{organization={Department of Industrial and Systems Engineering, University of Minnesota},
                city={Minneapolis},
                state={MN},
                country={USA}}

\affiliation[4]{organization={Institute for Sustainability, Energy, and Environment, University of Illinois Urbana-Champaign},
                city={Urbana},
                state={IL},
                country={USA}}

\cortext[cor1]{Corresponding author}

\begin{abstract}
We study a two-node stochastic resource system operating over a finite horizon. Each node experiences uncertain supply and demand and is equipped with finite storage. The objective is to ensure that resource levels remain within prescribed limits with high probability. To this end, we formulate a chance-constrained capacity-design problem in which resources can be exchanged through a capacity-limited transport link. We characterize the minimum storage required at each node, derive the optimal transport policy, and quantify the trade-off between storage and transport capacities. Our results show the existence of a critical transport-capacity threshold that enables full risk pooling between the nodes. Moreover, this threshold decreases with the operating horizon, implying that full-pooling performance can be achieved with progressively smaller transport capacity over longer horizons.
\end{abstract}








\maketitle

\section{Introduction}
Many networked systems consist of nodes that locally generate and consume resources under uncertainty. To mitigate shortages and surplus losses, each node is equipped with storage, whose capacity must be carefully designed due to cost constraints. When nodes are connected via transport links, resources can be exchanged, leading to a joint design problem over storage and transport capacities. The objective is to ensure that node-level resources remain within prescribed limits with high probability over a finite horizon. This setting arises in supply chains, where transport redistributes goods across production and distribution stages \cite{song2020capacity,lee2005optimization}, and in power systems, where interconnections enable energy sharing under renewable uncertainty \cite{alma9917556284304856,igder2022service,smith2013us}.

In this article, we study the trade-off between local storage and transport capacity in a two-node stochastic resource system. We address three questions: (i) the minimum storage required at each node to ensure capacity constraints are met with high probability over a finite horizon under stochastic supply and demand; (ii) the extent to which a transport link can reduce this storage requirement and the optimal transport policy for a given capacity; and (iii) the scaling laws governing the trade-off between storage, transport capacity, and operating horizon.

\subsection{Related Work}
Since uncertainty in production and consumption is often modeled by Brownian motion, ensuring that stored resources remain within prescribed limits with high probability naturally leads to stochastic control problems. A classical line of work studies Brownian motion with controlled drift. Benes~\cite{benevs1974girsanov} shows that optimal drift control for minimizing expected terminal cost in one dimension has a bang-bang structure, later extended to more general cost settings in~\cite{benevs1975composition,ikeda1977comparison,davis1979predicted,benevs1980some,christopeit1982benevs,chen2023bang,karatzas1993finite}. Related first-passage and reach-avoid problems are studied in~\cite{katayama1971stochastic,katayama1972optimal}, where bang-bang controls are again shown to be optimal, and optimal stopping formulations are considered in~\cite{karatzas2002leavable}. Motivated by this structure, transition properties of bang-bang drift processes are characterized in~\cite{shreve1981reflected}, with connections to reflected Brownian motion developed in~\cite{peskir1999maximal,graversen2000extension,peskir2006reflecting}, and further appearing in stochastic reachability analysis~\cite{abate2008probabilistic,gleason2019maximal}. However, these works do not address explicit optimal policies for the transport-coupled capacity design problem considered in this article.

In summary, existing work on stochastic control of Brownian motion largely focuses on scalar processes, while multidimensional extensions typically consider independent dynamics. Consequently, they do not address systems in which stochastic processes are coupled through resource transport. Motivated by applications that jointly exploit storage and transport to manage uncertainty, this article studies capacity design for interacting stochastic resource processes. Specifically,\\
1) For a connected two-node system with independent Brownian uncertainty at each node, we show that the storage capacity required to avoid curtailment and/or surplus losses over $[0,T_f]$ with high probability scales as $\sqrt{T_f}$. Further, we prove that transport can reduce the storage requirement only to a strictly positive lower bound, regardless of the transport capacity.\\
2) We prove that if this minimum storage level is maintained over increasing horizons, then the transport capacity can be decreased and designed proportionally to $1/\sqrt{T_f}$.\\
3) We derive the optimal transport feedback policy for any fixed transport capacity and show that it is bang-bang.

The remainder of the paper is organized as follows. Sections~\ref{sec:system_model} and~\ref{sec:problem_formulation} present the model and problem formulation. Section~\ref{sec:solution_methodology} develops the main results, including benchmark storage limits, the optimal transport policy, storage--transport scaling laws, and a finite-horizon design procedure. We conclude with key findings and design insights.
\section{System Model}\label{sec:system_model}
Consider a two-node facility system, shown in Fig.~\ref{fig:two_node_system}, where each node maintains a local resource and the nodes exchange resources via a transport link. Let $T_f$ denote the operating horizon, $\resourceCapacity$ the storage capacity at each node $i\in\{1,2\}$, and $\transport{1}{2}(t)$ the flow from node $1$ to node $2$ at time $t\in[0,T_f]$. Since the link is bidirectional, $\transport{2}{1}(t)=-\transport{1}{2}(t)$ with $|\transport{1}{2}(t)|\leq \transportCapacity$ for all $t\in[0,T_f]$, where $\transportCapacity$ denotes the transport capacity. Let $X_i(t)$ denote the resource level at node $i\in\{1,2\}$. Production and demand uncertainties induce stochastic evolution of $X_i(t)$, which we model as follows:
\begin{equation}\label{eq:brownian_resource_model_two_mg}
    \begin{aligned}
    X_1(t) &= \textstyle  X_1(0)+\sigma W_1(t)-\int_0^t \transport{1}{2}(s)\text{d}s, \\
    X_2(t) &=  \textstyle  X_2(0)+\sigma W_2(t)+\int_0^t \transport{1}{2}(s)\text{d}s,
    \end{aligned}
\end{equation}
where $X_i(0)$ is the initial resource at node $i$, $W_i(t)$ are independent Wiener processes modeling uncertainty, and $\sigma$ is the volatility parameter. The term $\int_0^t \transport{1}{2}(s)\mathrm{d}s$ denotes the cumulative transfer from node $1$ to node $2$ by time $t$. If $X_i(t)=\resourceCapacity$, excess production is curtailed, while $X_i(t)=0$ results in unmet demand. The objective is to choose $\resourceCapacity$ and the transport policy $\transport{1}{2}(t)$ such that $X_i(t)\in(0,\resourceCapacity)$ for both nodes over $[0,T_f]$ with high probability. The next section formalizes the problem.
\section{Problem Formulation}\label{sec:problem_formulation}
Let $X_{i,T_f}^{sup}\mspace{-5mu}:=\mspace{-5mu}\sup_{t\in[0,T_f]}X_i(t)$ and $X_{i,T_f}^{inf}\mspace{-5mu}:=\mspace{-5mu} \inf_{t\in[0,T_f]}X_i(t),$ $ i\mspace{-4mu}\in\mspace{-4mu}\{1,2\}$, 
denote the maximum and minimum resource levels at node $i\in \{1,2\}$ over the horizon $[0,T_f]$. Given $\delta\in(0,1)$, we require both node processes to remain in the safe region $(0,\resourceCapacity)$ with probability at least $1-\delta$, i.e.,
\begin{align}\label{eq:chance_constraint_two_mg}
    \textstyle \mathbb{P}[\cap_{i=1}^{2} \{X_{i,T_f}^{sup}\mspace{-4mu}<\mspace{-4mu} \resourceCapacity, \textstyle X_{i,T_f}^{inf}\mspace{-4mu}>\mspace{-4mu} 0\}] \geq\mspace{-4mu} 1\mspace{-4mu}-\mspace{-4mu}\delta.
\end{align}
Our objective is to find the minimum resource capacity $\resourceCapacity$ satisfying~\eqref{eq:chance_constraint_two_mg}. Let $\pi(t)=\transport{1}{2}(t)$ denote the transport policy and define the admissible policy set $\Pi(\transportCapacity)\mspace{-4mu}:=\mspace{-4mu}\{\pi: \sup_{0\leq t \leq T_f}|\pi(t)|\mspace{-4mu}\leq\mspace{-4mu} \transportCapacity\}$. We consider the initial resource allocation $X_i(0)=\alpha\resourceCapacity$, where $\alpha\in[0,1]$. Thus, given $\transportCapacity$, $T_f$, and $\delta \in (0,1)$, we solve
\begin{align}\label{eq:problem_statement_two_mg}
\textstyle \min_{\alpha, \resourceCapacity, \pi} & \textstyle \hspace{0.1in} \resourceCapacity \\
  & \hspace{-1cm}\text{subject to} \ (\ref{eq:brownian_resource_model_two_mg}),\ \textstyle \resourceCapacity \geq 0,\ \textstyle 0 \leq \alpha \leq 1, \pi \in \Pi(\transportCapacity),\nonumber\\
  & \hspace{-1cm}\textstyle \mathbb{P}[\cap_{i=1}^{2} \{X_{i,T_f}^{sup}\mspace{-4mu}<\mspace{-4mu} \resourceCapacity, \textstyle X_{i,T_f}^{inf}\mspace{-4mu}>\mspace{-4mu} 0\}] \geq\mspace{-4mu} 1\mspace{-4mu}-\mspace{-4mu}\delta\nonumber.
\end{align}
The solution determines the storage capacity, initial resource level, and transport policy used over $[0,T_f]$. In the following sections, we describe the solution methodology.
\begin{figure}
     \centering
     \includegraphics[scale=0.275,trim={0.3cm 7.4cm 3.3cm 4.9cm},clip]{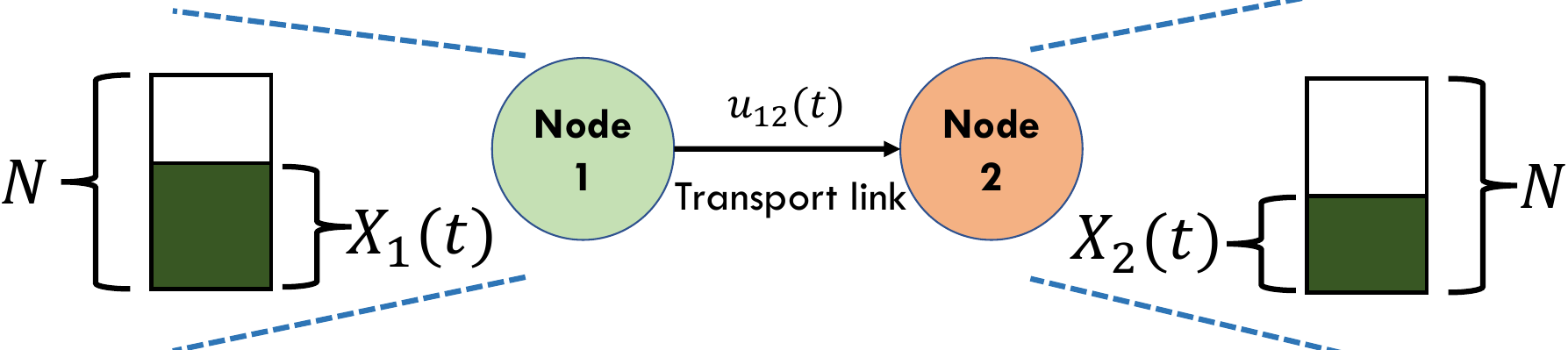}
     \caption{Two-node system model.}
     \label{fig:two_node_system}
\end{figure}
\section{Solution Methodology}\label{sec:solution_methodology}
Towards solving~\eqref{eq:problem_statement_two_mg}, we first provide some fundamental limits on the resource capacity to establish some benchmarks that clarify the role of resource sharing.
\subsection{Fundamental Limits}
First, we analyze a single stochastic resource without transport interaction and derive the capacity required to keep it within its admissible limits over $[0,T_f]$ with high probability. This result serves as a reference for the two-node setting. We then use it to study two limiting interaction regimes: zero transport capacity, where the nodes behave as independent processes, and infinite transport capacity, where the two processes can be fully aggregated. These cases establish fundamental benchmarks on the storage savings achievable through transport.

\subsubsection{Single process}

For a single process $X(t)$, we model the dynamics as $X(t)=\alpha \resourceCapacity+\sigma W(t)$, where $W(t)$ is a standard Wiener process. Here we characterize the minimum resource capacity required to ensure that $X(t)$ remains within the safe region $(0,\resourceCapacity)$ over the time horizon $[0,T_f]$ with high probability. Denoting $X_{T_f}^{sup} := \sup_{0 \leq t \leq T_f}X(t), X_{T_f}^{inf}:= \inf_{0 \leq t \leq T_f}X(t)$, the problem is as follows: Given $T_f, \delta \in (0,1)$,
\begin{align}\label{eq:problem_statement_single_mg}
\textstyle \min_{\alpha, \resourceCapacity} & \textstyle \hspace{0.1in} \resourceCapacity \\
  & \hspace{-1cm}\text{subject to} \ X(t)=\alpha \resourceCapacity + \sigma W(t),\ \textstyle \resourceCapacity \geq 0,\ \textstyle 0 \leq \alpha \leq 1,\nonumber\\
  & \hspace{-1cm}\textstyle \mathbb{P}[X_{T_f}^{sup}\mspace{-4mu}<\mspace{-4mu} \resourceCapacity, \textstyle X_{T_f}^{inf}\mspace{-4mu}>\mspace{-4mu} 0] \geq\mspace{-4mu} 1\mspace{-4mu}-\mspace{-4mu}\delta\nonumber.
\end{align}
Our earlier work~\cite{dey2024guaranteeingACC} establishes that the solution to~\eqref{eq:problem_statement_single_mg} is given by $(\alpha^\star, N^\star) := (1/2,2\sqrt{2}\sigma \sqrt{\ln(2/\delta)}\sqrt{T_f})$. Next, we consider the limiting cases of interaction between the two nodes.
\subsubsection{Limiting cases of interaction}\label{subsec:limiting_cases}
We consider two benchmark transport regimes: no interaction, $\transportCapacity=0$, and perfect interaction, $\transportCapacity\to\infty$. For each case, we derive the minimum resource capacity required to satisfy the chance constraint~\eqref{eq:chance_constraint_two_mg}. These limits quantify the maximum storage reduction achievable through transport. The chance constraint~\eqref{eq:chance_constraint_two_mg} can be equivalently written as:
\begin{align}\label{eq:chance_constraint_multi_mg_equivalent}
    & \textstyle \mathbb{P}[\cup_{i=1}^{2} \{\{X_{i,T_f}^{sup} \geq \resourceCapacity\} \cup \{X_{i,T_f}^{inf} \leq 0\}\}] \leq \delta.
\end{align}
When $\transportCapacity=0$, the processes do not interact and evolve independently as $X_i(t)=X_i(0) + \sigma W_i(t), i\in \{1,2\}$. Since $W_1(t)$ and $W_2(t)$ are independent of each other, we have
$ \mathbb{P}[\cup_{i=1}^{2}\{\{X_{i,T_f}^{sup}\geq \resourceCapacity\}\}\cup \{\{X_{i,T_f}^{inf}\leq 0\}\}] = \textstyle \mathbb{P}[\{X_{1,T_f}^{sup}\geq \resourceCapacity\}\cup \{X_{1,T_f}^{inf}\leq 0\}] + \mathbb{P}[\{X_{2,T_f}^{sup}\geq \resourceCapacity\}\cup \{X_{2,T_f}^{inf}\leq 0\}]. $ Since the processes are independent and identically distributed, we can satisfy~\eqref{eq:chance_constraint_two_mg} by treating each process separately and choosing the minimum storage capacity $\resourceCapacity$ for each node such that
$ \mathbb{P}[\{X_{1,T_f}^{sup}\geq N\}\cup \{X_{1,T_f}^{inf}\leq 0\}] \leq \delta/2, \text{ and, } \mathbb{P}[\{X_{2,T_f}^{sup}\geq N\}\cup \{X_{2,T_f}^{inf}\leq 0\}] \leq \delta/2$
thus satisfying~\eqref{eq:chance_constraint_multi_mg_equivalent} in unison. Let $\resourceCapacity^{(0)}$ denote the needed resource capacity. Following the analysis of single process and replacing $\delta$ with $\delta/2$ in the solution of~\eqref{eq:problem_statement_single_mg}, we can obtain $\resourceCapacity^{(0)} := 2\sqrt{2}\sigma\sqrt{\ln(4/\delta)}\sqrt{T_f}$.

When $\transportCapacity\to\infty$, the two processes can be fully pooled and treated as a single aggregate process. With $X_i(0)=\alpha\resourceCapacity$, define
$X_c(t):=X_1(t)+X_2(t)=2\alpha\resourceCapacity+\sigma(W_1(t)+W_2(t))$. Since $W_1(t)+W_2(t)\sim \sqrt{2}W_c(t)$~\cite{karatzas2012brownian}, where $W_c(t)$ is a Wiener process, we obtain $X_c(t)=2\alpha\resourceCapacity+\sigma_c W_c(t), \ \sigma_c:=\sqrt{2}\sigma.$ Because the aggregate process $X_c(t)$ is independent of the transport policy, the required per-node capacity $\resourceCapacity^{(\infty)}$ can be obtained from the single-process analysis applied to $X_c(t)$.
Here, replacing $\sigma$ with $\sigma_c$ in the solution of~\eqref{eq:problem_statement_single_mg}, we obtain $\resourceCapacity^{(\infty)} := 2\sigma\sqrt{\ln(2/\delta)}\sqrt{T_f}$.

In summary, infinite transport capacity reduces the minimum resource capacity by a factor $\frac{\resourceCapacity^{(0)}}{\resourceCapacity^{(\infty)}}=\sqrt{\frac{2\ln(4/\delta)}{\ln(2/\delta)}}$ relative to the no-transport case. Since higher transport capacity entails greater infrastructure cost, a natural question is whether the chance constraint~\eqref{eq:chance_constraint_multi_mg_equivalent} can be satisfied with the minimal resource capacity $\resourceCapacity^{(\infty)}$ using only finite transport capacity. We address this question next.

\subsection{Optimal Transport Policy Under Specified Transport and Storage Capacities}
In this section, we focus on characterizing the optimal transport policy for fixed resource capacity $\resourceCapacity$ and transport capacity $\transportCapacity$. In particular, given $\resourceCapacity,\transportCapacity$, $X_i(0)=x_i \in (0,\resourceCapacity)$, and horizon $T_f$, we solve:
\begin{align}\label{eq:prob_maximization_two_mg}
    & \textstyle \max_{\pi\in \Pi(\transportCapacity)} \mathbb{P}[\cap_{i=1}^{2} \{X_{i,T_f}^{sup}\mspace{-4mu}<\mspace{-4mu} \resourceCapacity, \textstyle X_{i,T_f}^{inf}\mspace{-4mu}>\mspace{-4mu} 0\}]\nonumber\\
    & \text{s. t. } X_i \text{ follows }\eqref{eq:brownian_resource_model_two_mg} \text{ with } X_i(0)=x_i, i \in \{1,2\}.
\end{align}
This problem is closely related to the power-sharing problem studied in our earlier work~\cite{dey2024guaranteeingACC}, where a dynamic programming approach was used to derive the optimal transfer policy for two interacting stochastic storage processes under a fixed line capacity. That analysis shows that the optimal policy has a \textit{bang-bang} structure: at each time, the node with higher resource transfers to the one with lower resource using full transport capacity. Translating this to the present setting, the optimal policy for~\eqref{eq:prob_maximization_two_mg} is

\begin{align}\label{eq:optimal_policy_bang_bang}
\pi^\star(t) = \transport{1}{2}^\star(t) = 
\begin{cases}
\transportCapacity, & \mbox{if} \ X_1(t)>X_2(t),\\
-\transportCapacity, & \mbox{if} \ X_1(t)<X_2(t),\\
0, & \mbox{if} \ X_1(t) = X_2(t).
\end{cases}
\end{align}
Thus, for any fixed transport capacity $\transportCapacity$, the optimal strategy is to fully utilize the link in the direction that reduces the resource imbalance between the two nodes. Since transport capacity contributes to infrastructure cost, we next study the storage–transport trade-off and characterize how finite transport capacity can be chosen to satisfy~\eqref{eq:chance_constraint_two_mg} while maintaining the minimum storage level $\resourceCapacity^{(\infty)}$.
\subsection{Critical Transport Capacity}\label{subsec:capacity_trade_offs}
Since $\resourceCapacity^{(\infty)}$ is the minimum storage capacity achievable under perfect pooling, it is desirable from a cost perspective. However, higher transport capacity increases infrastructure cost, leading to the question of whether the chance constraint~\eqref{eq:chance_constraint_multi_mg_equivalent} can still be satisfied with storage fixed at $\resourceCapacity^{(\infty)}$ using only finite transport capacity. We address this by characterizing the critical transport-capacity scaling required to maintain feasibility at the minimal storage level, under an asymptotic regime where the operating horizon grows without bound.

Suppose each node has resource capacity $\resourceCapacity^{(\infty)}$. Let the transport capacity scale as $\transportCapacity=c_pT_f^{\xi}$, where $\xi\in\mathbb{R}$ and $c_p>0$ is independent of $T_f$. Let $X_1^{\pi^\star}(t)$ and $X_2^{\pi^\star}(t)$ denote the solutions to~\eqref{eq:brownian_resource_model_two_mg} under the optimal bang-bang transport policy $\transport{1}{2}^\star$ defined in~\eqref{eq:optimal_policy_bang_bang}, namely,
\begin{equation}\label{eq:brownian_soc_two_mg_under_optimal_policy}
    \begin{aligned}
        \textstyle X_1^{\pi^\star}(t) &\mspace{-4mu}=\mspace{-4mu} \textstyle \frac{1}{2} \resourceCapacity^{(\infty)} \mspace{-4mu}+\mspace{-4mu} \sigma W_1(t) \mspace{-4mu}-\mspace{-4mu} \B\int_0^t \transport{1}{2}^\star(s)\text{d}s, \\
        \textstyle X_2^{\pi^\star}(t) &\mspace{-4mu}=\mspace{-4mu} \textstyle \frac{1}{2} \resourceCapacity^{(\infty)} \mspace{-4mu}+\mspace{-4mu} \sigma W_2(t) \mspace{-4mu}+\mspace{-4mu} \B\int_0^t \transport{1}{2}^\star(s)\text{d}s.
    \end{aligned}
\end{equation}
Let $X_{i,T_f}^{\pi^\star sup} := \sup_{0\leq t \leq T_f}X_i^{\pi^\star}(t), X_{i,T_f}^{\pi^\star inf} := \inf_{0 \leq t \leq T_f} $ $ X_i^{\pi^\star}(t), i \in \{1,2\}$, and the following quantities for the resource capacity $\resourceCapacity^{(\infty)}$:
\begin{equation}\label{eq:set_notation}
    \begin{aligned}
    \textstyle A_1 &:= \textstyle \{X_{1,T_f}^{\pi^\star sup} \geq \resourceCapacity^{(\infty)}\},\ \textstyle A_2 := \textstyle \{X_{1,T_f}^{\pi^\star inf}\leq 0\} \\
    \textstyle A_3 &:= \textstyle \{X_{2,T_f}^{\pi^\star sup} \geq \resourceCapacity^{(\infty)}\},\ \textstyle A_4 := \textstyle \{X_{2,T_f}^{\pi^\star inf}\leq 0\}.
\end{aligned}
\end{equation}
Then, the chance-constraint~\eqref{eq:chance_constraint_multi_mg_equivalent} can be written as
\begin{align}\label{eq:two_mg_chance_constraint_equivalent_n_infinity}
    \textstyle \mathbb{P}[A_1\cup A_2 \cup A_3\cup A_4] \leq \delta.
\end{align}

A smaller value of $\delta$ imposes a stricter reliability requirement by reducing the admissible probability of surplus or curtailment. We show that for $\xi<-\tfrac{1}{2}$, the chance constraint~\eqref{eq:two_mg_chance_constraint_equivalent_n_infinity} cannot be satisfied for sufficiently small $\delta$ over large horizon $T_f$. Consequently, maintaining the storage level $\resourceCapacity^{(\infty)}$ for all reliability thresholds requires that the transport capacity does not decay faster than order $1/\sqrt{T_f}$. This necessary condition is formalized in the following theorem.
\begin{theorem}\label{thm:lower_bound}
    Let $X_1^{\pi^\star}(t)$ and $X_2^{\pi^\star}(t)$ be as in~\eqref{eq:brownian_soc_two_mg_under_optimal_policy} with the resource capacity $\resourceCapacity^{(\infty)}$ and $A_1,A_2,A_3,A_4$ be as defined in~\eqref{eq:brownian_soc_two_mg_under_optimal_policy}. Let the transport capacity be, $\transportCapacity = c_p T_f^{\xi}$, where $\xi < -\frac{1}{2}$, and $c_p$ is a positive constant independent of $T_f$. There exists a $\delta^* \in (0,1)$, such that for all $0 < \delta \leq \delta^*$, $\lim_{T_f \to \infty} \mathbb{P}[A_1\cup A_2 \cup A_3\cup A_4] > \delta$.
\end{theorem}
\begin{proof}
    Please see Appendix~\ref{appendix:lower_bound}.
\end{proof}
Having established the necessary lower bound on the decay rate of the transport capacity, we now show that this rate is essentially sufficient. Specifically, if $\transportCapacity=c_pT_f^{\xi}$ with $\xi\in(-\tfrac{1}{2},0)$, then the chance constraint~\eqref{eq:two_mg_chance_constraint_equivalent_n_infinity} is satisfied for any $\delta\in(0,1)$ over large $T_f$, while maintaining the storage level $\resourceCapacity^{(\infty)}$ at each node. This implies that the transport capacity need not remain constant with growing horizon; it may decay with $T_f$, provided the decay is slower than $1/\sqrt{T_f}$. Hence, $\transportCapacity=\Theta(T_f^{-1/2})$, (Big Theta notation~\cite{knuth1976big}), emerges as the critical threshold separating feasibility and infeasibility of full-pooling performance with finite transport capacity. The result is formalized in the following theorem.
\begin{theorem}\label{thm:upper_bound}
    Let $X_1^{\pi^\star}(t)$ and $X_2^{\pi^\star}(t)$ be as in~\eqref{eq:brownian_soc_two_mg_under_optimal_policy} with the resource capacity $\resourceCapacity^{(\infty)}$ and $A_1,A_2,A_3,A_4$ be as defined in~\eqref{eq:brownian_soc_two_mg_under_optimal_policy}. Let the transport capacity be $\transportCapacity = c_p T_f^{-\frac{1}{2}+\epsilon}$, where $0 < \epsilon < \frac{1}{2}$, and $c_p$ is a positive constant independent of $T_f$. Given $\delta \in (0,1)$ we have,
    \begin{align*}
        \textstyle \lim_{T_f \to \infty} \mathbb{P}[A_1 \cup A_2 \cup A_3 \cup A_4] \leq \delta.
    \end{align*}
\end{theorem} 
Theorem~\ref{thm:upper_bound} has important infrastructure-planning implications. While the minimum storage requirement $\resourceCapacity^{(\infty)}=2\sigma\sqrt{\ln(2/\delta)}\sqrt{T_f}$ grows with the operating horizon, the transport capacity needed to maintain this storage level can decrease at a rate approaching $1/\sqrt{T_f}$. Hence, longer horizons require greater storage but less transport capacity, highlighting a fundamental storage-transport trade-off in reliable stochastic resource systems.

We provide the proof of Theorem~\ref{thm:upper_bound} utilizing two intermediate: Proposition~\ref{prop:conservative_sol_two_mg} and Lemma~\ref{lem:difference_process_upper_bound}. We first present these results and then use them to prove Theorem~\ref{thm:upper_bound}. Let $X_{c}^{\pi^\star}(t):= X_1^{\pi^\star}(t)+X_2^{\pi^\star}(t)$ and $X_d^{\pi^\star}(t) := X_2^{\pi^\star}(t)-X_1^{\pi^\star}(t), t\in [0,T_f]$. Then, from~\eqref{eq:brownian_soc_two_mg_under_optimal_policy}, we have the following:
\begin{equation}\label{eq:brownian_soc_two_mg_1}
    \begin{aligned}
        \textstyle X_c^{\pi^\star}(t) & \textstyle = \resourceCapacity^{(\infty)}\mspace{-4mu}+\mspace{-4mu} \sigma (W_1(t) + W_2(t)) \\ 
        \textstyle X_d^{\pi^\star}(t) & \textstyle = \sigma (W_2(t)-W_1(t))+2\B\int_0^t \transport{1}{2}^\star(s)\text{d}s.
\end{aligned}
\end{equation}
We have $W_1(t)+W_2(t)\sim \sqrt{2}W_c(t)$, where $W_c(t)$ is a Wiener process~\cite{karatzas2012brownian}. Hence, $X_c^{\pi^\star}(t)$ is given as
\begin{align}\label{eq:sum_equation}
    \textstyle X_c^{\pi^\star}=\resourceCapacity^{(\infty)}+\sigma_c W_c(t),
\end{align}
where $\sigma_c=\sqrt{2}\sigma$. Similarly, $W_2(t)-W_1(t)\sim \sqrt{2}W_d(t)$, where $W_d(t)$ is a Wiener process~\cite{karatzas2012brownian}. Further, $\transport{1}{2}^\star(t)$ in~\eqref{eq:optimal_policy_bang_bang}, can be written as
\begin{align}\label{eq:optimal_power_transfer_policy_sgn}
    \textstyle \transport{1}{2}^\star(t) = -\transportCapacity\text{sgn}(X_d^{\pi^\star}(t)).
\end{align}
Thus, $X_d^{\pi^\star}(t)$ takes the form, with $\sigma_d:=\sqrt{2}\sigma$,
\begin{align}\label{eq:difference_equation}
    \textstyle X_d^{\pi^\star}=\sigma_d W_d(t)-2\B\transportCapacity\int_{0}^{t}\text{sgn}(X_d^{\pi^\star}(s))\text{d}s.
\end{align}
Note that $X_c^{\pi^\star}(t)$ does not depend on the transport capacity $\transportCapacity$, while $X_d^{\pi^\star}(t)$ captures the effect of transport through the bang-bang policy. Since both $X_c^{\pi^\star}(t)$ and $X_d^{\pi^\star}(t)$ are one-dimensional processes, chance constraint~\eqref{eq:two_mg_chance_constraint_equivalent_n_infinity} can be analyzed tractably through their behavior. To this end, define
\begin{align}\label{eq:sum_equation_sup_inf}
   \hspace{-0.35in} X_{c,T_f}^{\pi^\star sup} := \sup_{t\in[0,T_f]}X_c^{\pi^\star}(t),  X_{c,T_f}^{\pi^\star inf} := \inf_{t\in [0,T_f]}X_c^{\pi^\star}(t).
\end{align}

In the next proposition, we first show how the probability $\mathbb{P}[A_1 \cup A_2 \cup A_3 \cup A_4]$ can be upper bounded in terms of probabilities involving $X_{c,T_f}^{\pi^\star sup}, X_{c,T_f}^{\pi^\star inf}$ and $X_d^{\pi^\star}(t)$.
\begin{proposition}\label{prop:conservative_sol_two_mg}
    Let $ X_c^{\pi^\star}, X_d^{\pi^\star}(t)$, and $X_{c,T_f}^{\pi^\star sup}, X_{c,T_f}^{\pi^\star inf}$ be as  in~\eqref{eq:sum_equation},~\eqref{eq:difference_equation} and~\eqref{eq:sum_equation_sup_inf} respectively, and  $A_1,A_2,A_3,A_4$ be the same as in~\eqref{eq:set_notation}. For $0 < \beta < \resourceCapacity^{(\infty)}$, we have
    \begin{align*}
    \textstyle & \mathbb{P}[A_1 \cup A_2 \cup A_3 \cup A_4]\\
     & \textstyle \leq \mathbb{P}[X_{c,T_f}^{\pi^\star sup} \geq 2\resourceCapacity^{(\infty)}-\beta]\textstyle  + \mathbb{P}[X_{c,T_f}^{\pi^\star inf} \leq \beta]\nonumber\\
    & \hspace{0.1in} \textstyle + \mathbb{P}[\sup_{0\leq t \leq T_f}|X_{d}^{\pi^\star}(t)| \geq \beta].
\end{align*}
\end{proposition}
\begin{proof}
    Please see Appendix~\ref{appendix:conservative_sol_two_mg}.
\end{proof}
Proposition~\ref{prop:conservative_sol_two_mg} is central to the proof of Theorem~\ref{thm:upper_bound}, as it decomposes the original two-resource chance constraint into two components: the aggregate process $X_c^{\pi^\star}$ and the imbalance process $X_d^{\pi^\star}$. The former governs whether the total available resource is sufficient, while the latter captures its distribution across nodes and determines whether imbalance can lead to node-level violations.

As stated in Theorem~\ref{thm:upper_bound}, let $\transportCapacity=c_pT_f^{-\frac{1}{2}+\epsilon}$ with $\epsilon\in(0,\frac{1}{2})$, and fix the storage capacity at the level  $\resourceCapacity^{(\infty)}=2\sigma\sqrt{\ln(2/\delta)}\sqrt{T_f}$. We choose the intermediate threshold $\beta=T_f^{\frac{1}{2}-\frac{\epsilon}{2}}$, which grows with $T_f$ but more slowly than $\sqrt{T_f}$. We first analyze $\mathbb{P}[\sup_{0\le t\le T_f}|X_d^{\pi^\star}(t)|\ge \beta]$, the probability that the imbalance process exceeds $\beta$ over the horizon. We show that under the bang-bang policy,
\[
\lim_{T_f\to\infty}
\mathbb{P}\Big[\sup_{0\le t\le T_f}|X_d^{\pi^\star}(t)|\ge\beta\Big]=0.
\]
Hence, transport capacity scaling as $c_pT_f^{-\frac{1}{2}+\epsilon}$ is sufficient to keep the imbalance between the nodes small relative to the storage margin with high probability as $T_f$ grows. We present this analysis next. Note that for any $\beta$
\begin{align}\label{eq:diff_process_mod_inequality}
    & \hspace{-0.3in} \textstyle \mathbb{P}[\sup_{0\leq t \leq T_f}|X_{d}^{\pi^\star}(t)| \geq \beta] = \mathbb{P}\big[\{\sup_{0\leq t \leq T_f} X_{d}^{\pi^\star}(t) \geq \beta\} \nonumber \\
    & \hspace{-0.3in} \textstyle \cup \{\inf_{0\leq t \leq T_f} X_{d}^{\pi^\star}(t) \leq -\beta\}\big] \leq \mathbb{P}[\sup_{0\leq t \leq T_f} X_{d}^{\pi^\star}(t) \geq \beta] \nonumber\\
    & \hspace{0.9in} \textstyle  + \mathbb{P}[\inf_{0\leq t \leq T_f} X_{d}^{\pi^\star}(t) \leq -\beta],
\end{align}
where the last inequality follows from the union bound on probabilities. We show that as $T_f \to \infty$, both the probabilities on the right-hand side of~\eqref{eq:diff_process_mod_inequality} go to zero and hence, $\lim_{T_f \to \infty}\mathbb{P}[\sup_{0\leq t \leq T_f}|X_{d}^{\pi^\star}(t)| \geq \beta]=0$. We present the following Lemma.
\begin{lemma}\label{lem:difference_process_upper_bound}
    Let $X_d^{\pi^\star}(t)$ be as defined in~\eqref{eq:difference_equation}, and 
    let $\epsilon, \gamma$ be constants such that $0 < \epsilon < \frac{1}{2} - \gamma$. Let $\beta := T_f^{(\frac{1}{2}-\frac{\epsilon}{2})}$ and the transport capacity, $\transportCapacity = c_p T_f^{-\frac{1}{2}+\epsilon}$. We have 
    \begin{align*}
       &\hspace{-0.35in} \mathbb{P}[\sup_{0\leq t \leq T_f} X_{d}^{\pi^\star}(t) \geq \beta] \leq [8\B c_pT_f^{1-2\gamma+(\frac{\epsilon}{2})}+1]e^{-\frac{2\B c_p}{\sigma_d^2}[T_f^{\frac{\epsilon}{2}}-T_f^{-\frac{1}{2}+\gamma+\epsilon}]} \\
       &\hspace{-0.35in} \textstyle \mathbb{P}[\inf_{0\leq t \leq T_f} X_{d}^{\pi^\star}(t) \leq -\beta] \leq \\
       & \hspace{0.75in} \textstyle (8\B c_pT_f^{1-2\gamma+(\frac{\epsilon}{2})}+1)e^{-\frac{2\B c_p}{\sigma_d^2}(T_f^{\frac{\epsilon}{2}}-T_f^{-\frac{1}{2}+\gamma+\epsilon})}
    \end{align*}
In addition, we can conclude that 
    \begin{align*}
        \textstyle & \textstyle \lim_{T_f \to \infty}\mathbb{P}[\sup_{0\leq t \leq T_f} X_{d}^{\pi^\star}(t) \geq \beta]=0, \text{ and,}\\
        \textstyle & \textstyle \lim_{T_f \to \infty}\mathbb{P}[\inf_{0\leq t \leq T_f} X_{d}^{\pi^\star}(t) \leq -\beta]=0.
    \end{align*}
\end{lemma}
\begin{proof}
    Please see Appendix~\ref{appendix:difference_process_upper_bound}.
\end{proof}
Lemma~\ref{lem:difference_process_upper_bound} shows that, under the transport capacity choice $\transportCapacity=c_pT_f^{-\frac{1}{2}+\epsilon}$ with $0<\epsilon<\frac{1}{2}$, the bang-bang policy increasingly suppresses the resource imbalance as the operating horizon grows. Specifically, for the threshold $\beta=T_f^{\frac{1}{2}-\frac{\epsilon}{2}}$, the probabilities that the imbalance process $X_d^{\pi^\star}(t)$ exceeds $\beta$ or falls below $-\beta$ over $[0,T_f]$ both vanish as $T_f\to\infty$. Therefore, for $\beta=T_f^{\frac{1}{2}-\frac{\epsilon}{2}}$, combining these bounds with~\eqref{eq:diff_process_mod_inequality}, we obtain
\begin{align}\label{eq:difference_process_zero}
    \textstyle \lim_{T_f \to \infty}\mathbb{P}[\sup_{0\leq t \leq T_f}|X_{d}^{\pi^\star}(t)| \geq \beta]=0.
\end{align}

We now present a proof sketch of Theorem~\ref{thm:upper_bound} (The formal proof of Theorem~\ref{thm:upper_bound} is given in Appendix~\ref{appendix:upper_bound}). For the choice, $\beta=T_f^{\frac{1}{2}-\frac{\epsilon}{2}}$, we first show that the aggregate-process satisfies
\[
\lim_{T_f \to \infty}
\left(
\mathbb{P} \left[X_{c,T_f}^{\pi^\star sup}\geq 2\resourceCapacity^{(\infty)}-\beta\right]
+
\mathbb{P}\left[X_{c,T_f}^{\pi^\star inf}\leq \beta\right]
\right)
\leq \delta .
\]
This establishes that the aggregate resource remains within the safety limits with the prescribed probability as the time horizon grows. Combined with~\eqref{eq:difference_process_zero}, which guarantees that the imbalance process remains below $\beta$ with probability approaching one, and Proposition~\ref{prop:conservative_sol_two_mg}, it follows that the chance constraint is satisfied with storage capacity $\resourceCapacity^{(\infty)}$ and transport capacity $\transportCapacity=c_pT_f^{-\frac{1}{2}+\epsilon}$, $0<\epsilon<\frac{1}{2}$, that is,
\[
\lim_{T_f\to\infty}
\mathbb{P}[A_1\cup A_2\cup A_3\cup A_4]\leq \delta.
\]
Hence, the transport capacity may decay at any rate slower than $1/\sqrt{T_f}$ while preserving the storage level $\resourceCapacity^{(\infty)}$ and satisfying the chance constraint. 

\section{Simulation}
We validate the chance constraint in Theorem~\ref{thm:upper_bound} via Monte Carlo simulation of~\eqref{eq:brownian_soc_two_mg_under_optimal_policy} under the bang-bang policy~\eqref{eq:optimal_policy_bang_bang}. We fix $\sigma=2$, $\delta=0.01$, and $c_p=100$, and consider multiple horizon lengths $T_f$. The storage capacity is parameterized as $N = K N^{(\infty)}$, $K \in \mathbb{R}_+$, while the transport capacity is $\transportCapacity = c_p T_f^{-1/2+\epsilon}$. The dynamics of $X_1^{\pi^\star}(t)$ and $X_2^{\pi^\star}(t)$ are simulated using a discrete-time approximation. For each $\epsilon$, we compute the minimal feasible $K$ satisfying~\eqref{eq:two_mg_chance_constraint_equivalent_n_infinity}, estimated over $10^4$ sample paths. As shown in Fig.~\ref{fig:K_vs_epsilon_Tf10000}, a clear storage–transport trade-off is observed for fixed $T_f$: increasing transport capacity (larger $\epsilon$) reduces the required storage (smaller $K$), and vice versa. Fig.~\ref{fig:K_vs_epsilon_Tf_convergence} further shows that as $T_f$ increases, this trade-off goes towards a critical transition with the feasible region contracting toward $K \to 1$ and $\epsilon \to 0$, consistent with Theorem~\ref{thm:upper_bound}.
\begin{figure}
     \centering
     \includegraphics[scale=0.4,trim={0.3cm 0.4cm 0.0cm 0.2cm},clip]{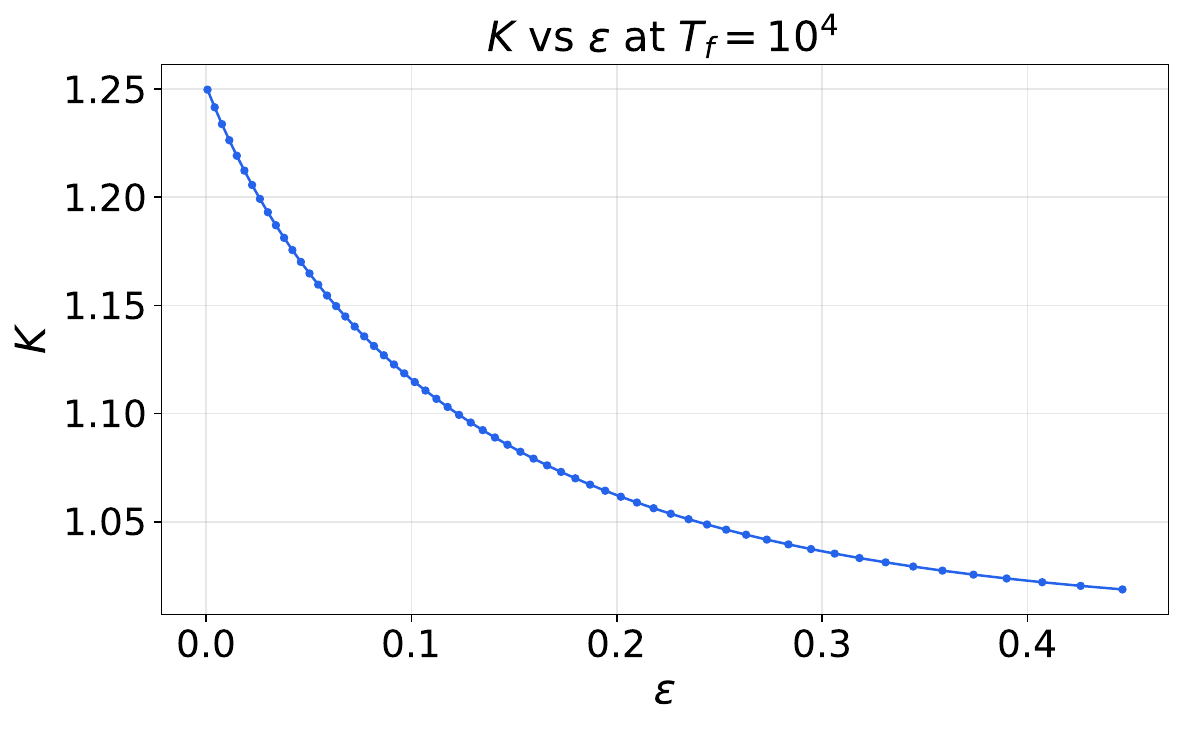}
     \caption{Tradeoff between storage and transport capacities for $T_f=10^4$. The storage capacity is $N=K N^{(\infty)}$ and the transport capacity is $U=c_p T_f^{-1/2+\epsilon}$. The tradeoff reflects compensation between storage and transport resources to satisfy the chance constraint.}
     \label{fig:K_vs_epsilon_Tf10000}
\end{figure}
\begin{figure}
     \centering
     \includegraphics[scale=0.45,trim={0.3cm 0.1cm 0.0cm 0.2cm},clip]{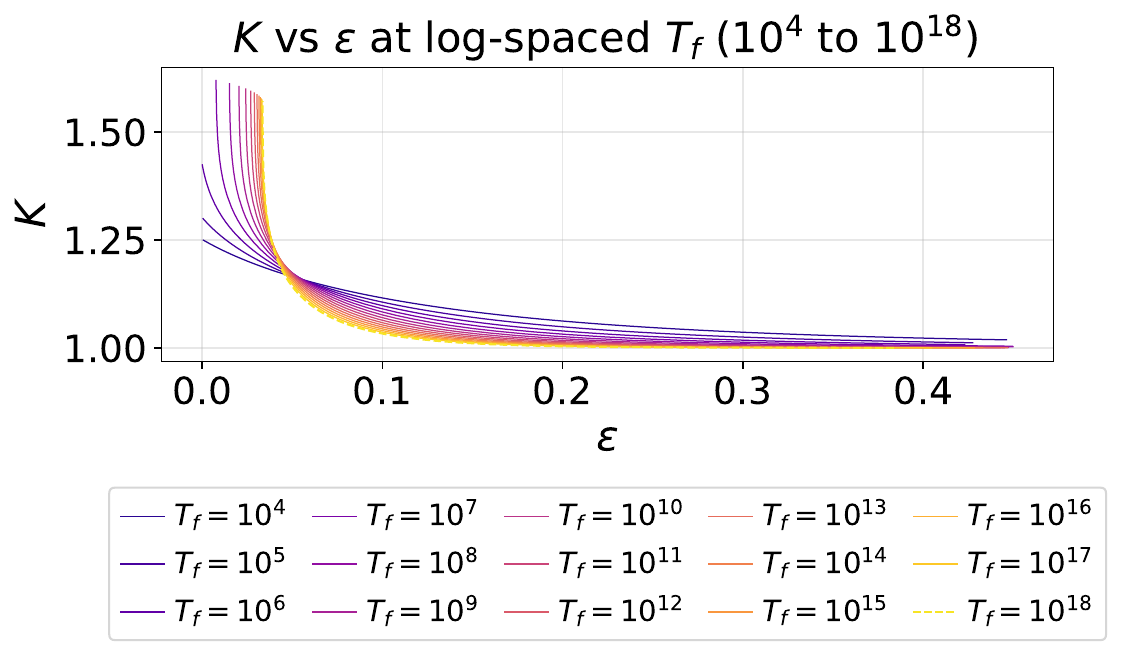}
     \caption{Convergence of the feasible storage-transport capacity as $T_f$ increases. Increasing $T_f$ allows feasibility to be achieved with $K \to 1$ and $\epsilon \to 0$ while satisfying the chance constraint.}
     \label{fig:K_vs_epsilon_Tf_convergence}
\end{figure}
\section{Conclusion}
We study capacity design in a two-node stochastic resource system connected by a capacity-limited transport link. Resources evolve as Brownian-motion-driven processes capturing uncertain supply and demand, with the objective of maintaining both within prescribed limits over a finite horizon with high probability. We characterize the storage–transport trade-off. Without transport, storage scales as $\sqrt{T_f}$, while with unlimited transport the system achieves full pooling, reducing storage to a strictly positive limit. The optimal transport policy is bang-bang, transferring at full capacity from the higher to the lower node. Using a sum and difference decomposition, we show that the minimum storage capacity achieved under perfect pooling is achievable if transport capacity decays slower than $1/\sqrt{T_f}$, while faster decay is infeasible. These results provide design principles for stochastic resource networks exploiting both temporal and spatial diversity.

\appendix
\section{Proof of Theorem~\ref{thm:lower_bound}}\label{appendix:lower_bound}
\vspace{-0.2in}
\begin{align}\label{eq:lower_bound_inequality_1}
   \hspace{-0.35in} \mbox{Note that} \ \textstyle \mathbb{P}[A_1 \cup A_2 \cup A_3 \cup A_4] & \textstyle \geq \mathbb{P}[A_1]\nonumber\\
    & \hspace{-0.65in} \textstyle = \mathbb{P}[\sup_{0 \leq t \leq T_f} X_1^{\pi^\star}(t) \geq \resourceCapacity^{(\infty)}].
\end{align}
For all $s\in [0,T_f]$, $-\transportCapacity\leq \transportCapacity\sgn(X_d^{\pi^\star}(s)) \leq \transportCapacity,$ 
and hence, $\frac{1}{2} \resourceCapacity^{(\infty)}+\sigma W_1(t) -\B\transportCapacity t \leq \frac{1}{2} \resourceCapacity^{(\infty)}+\sigma W_1(t)+\B\transportCapacity\int_{0}^{t} \sgn(X_d^{\pi^\star}(s))\text{d}s \leq \frac{1}{2} \resourceCapacity^{(\infty)}+\sigma W_1(t)+\B\transportCapacity t,$
for all $t \in [0, T_f]$ which implies $\frac{1}{2} \resourceCapacity^{(\infty)}+\sigma W_1(t) -\B\transportCapacity t \leq X_1^{\pi^\star}(t) \leq \frac{1}{2} \resourceCapacity^{(\infty)}+\sigma W_1(t)+\B\transportCapacity t.$
Therefore, $\sup_{0\leq t \leq T_f} (\frac{1}{2} \resourceCapacity^{(\infty)}+\sigma W_1(t) -\B\transportCapacity t) \geq \resourceCapacity^{(\infty)} \implies \sup_{0\leq t \leq T_f} X_1^{\pi^\star}(t) \geq \resourceCapacity^{(\infty)}$.
\begin{align*}
   \hspace{-0.35in} \mbox{Thus,} \ \ & \textstyle \mathbb{P}[\sup_{0\leq t \leq T_f} X_1^{\pi^\star}(t) \geq \resourceCapacity^{(\infty)}]\\
    & \textstyle \geq \mathbb{P}[\sup_{0\leq t \leq T_f} (\frac{1}{2} \resourceCapacity^{(\infty)}+\sigma W_1(t) -\B\transportCapacity t) \geq \resourceCapacity^{(\infty)}]\\
    & \textstyle =\mathbb{P}[\sup_{0\leq t \leq T_f} (W_1(t) -\frac{\B\transportCapacity}{\sigma} t) \geq \frac{\resourceCapacity^{(\infty)}}{2\sigma}]
\end{align*}
Note that, $(W_1(t) -\frac{\B\transportCapacity}{\sigma} t)$ is a Brownian motion with constant drift. A standard inequality involving the hitting time of such drifted Brownian motion can be obtained using Girsanov's theorem \cite{karatzas2012brownian,gordon1941values}, as follows:
\begin{align*}
    \textstyle & \textstyle \mathbb{P}[\sup_{0\leq t \leq T_f} (W_1(t) -\frac{\B\transportCapacity}{\sigma} t) \geq \frac{\resourceCapacity^{(\infty)}}{2\sigma}]\\
    & \textstyle \geq e^{-\frac{\B\transportCapacity\resourceCapacity^{(\infty)}}{2\sigma^2}-\frac{1}{2}\frac{\B \transportCapacity^2}{\sigma^2}T_f} \frac{2}{\sqrt{2\pi}}\frac{(\frac{\resourceCapacity^{(\infty)}}{2\sigma \sqrt{T_f}})}{1+(\frac{\resourceCapacity^{(\infty)}}{2\sigma \sqrt{T_f}})^2}e^{-\frac{1}{2}(\frac{\resourceCapacity^{(\infty)}}{2\sigma \sqrt{T_f}})^2}
\end{align*}
\noindent Since, $\transportCapacity=c_p T_f^{\xi}, \resourceCapacity^{(\infty)} = 2\sigma\sqrt{\ln(2/\delta)}\sqrt{T_f} =: c_{\infty}\sqrt{T_f}$,
\begin{align}\label{eq:lower_bound_inequality_7}
    & \textstyle \mathbb{P}[\sup_{0 \leq t \leq T_f} X_1^{\pi^\star}(t) \geq \resourceCapacity^{(\infty)}]\nonumber\\
    & \textstyle \geq e^{-\frac{\B c_pc_{\infty}}{2\sigma^2}T_f^{\frac{1}{2}+\xi}-\frac{c_p^2}{2\sigma^2} T_f^{2\xi+1}} \frac{2}{\sqrt{2\pi}}\frac{(\frac{c_{\infty}}{2\sigma})}{1+(\frac{c_{\infty}}{2\sigma})^2}e^{-\frac{1}{2}(\frac{c_{\infty}}{2\sigma})^2}.
\end{align}
As $\xi < -\frac{1}{2}$, we have $\xi+\frac{1}{2} < 0$, and also $2\xi+1 < 0$. Then, 
$\lim_{T_f \to \infty} e^{-\frac{\B c_pc_{\infty}}{2\sigma^2}T_f^{\frac{1}{2}+\xi}-\frac{c_p^2}{2\sigma^2} T_f^{2\xi+1}} = 1.$
Thus, we have $\lim_{T_f \to \infty} e^{-\frac{\B c_pc_{\infty}}{2\sigma^2}T_f^{\frac{1}{2}+\xi}-\frac{c_p^2}{2\sigma^2} T_f^{2\xi+1}} \frac{2}{\sqrt{2\pi}}\frac{(\frac{c_{\infty}}{2\sigma})}{1+(\frac{c_{\infty}}{2\sigma})^2}e^{-\frac{1}{2}(\frac{c_{\infty}}{2\sigma})^2} = \frac{2}{\sqrt{2\pi}}\frac{(\frac{c_{\infty}}{2\sigma})}{1+(\frac{c_{\infty}}{2\sigma})^2}e^{-\frac{1}{2}(\frac{c_{\infty}}{2\sigma})^2} = \frac{2}{\sqrt{2\pi}}\frac{\sqrt{\ln(2/\delta)}}{1+\ln(2/\delta)}\sqrt{\frac{\delta}{2}}$.
There exists a $\delta^*\in(0,1)$ such that for all $\delta \leq \delta^*$, $\frac{2}{\sqrt{2\pi}}\frac{\sqrt{\ln(2/\delta)}}{1+\ln(2/\delta)}\sqrt{\frac{\delta}{2}}-\delta > 0$.
Therefore, $\forall \delta \leq \delta^*$, $\lim_{T_f \to \infty}\mathbb{P}[\sup_{0\leq t \leq T_f} X_1^{\pi^\star}(t) \geq \resourceCapacity^{(\infty)}] > \delta$ and thus we have the desired result.
\section{Proof of Proposition~\ref{prop:conservative_sol_two_mg}}\label{appendix:conservative_sol_two_mg}
Let $A_d := \{\sup_{0 \le t \le T_f}|X_d^{\pi^\star}(t)| < \beta\}$ and $A_d' := \{\sup_{0 \leq t \leq T_f} |X_{d}^{\pi^\star}(t)| \geq \beta\}$. Then
$\mathbb{P}(A_1 \cup A_2 \cup A_3 \cup A_4)
= \mathbb{P}((A_1 \cup A_2 \cup A_3 \cup A_4)\cap A_d)
+ \mathbb{P}((A_1 \cup A_2 \cup A_3 \cup A_4)\cap A_d')$.

On $A_d$, using standard set algebra (with intermediate steps omitted for brevity), we obtain $A_1\cap A_d \subseteq A_c, A_3\cap A_d \subseteq A_c$ and $A_2\cap A_d \subseteq \tilde A_c, A_4\cap A_d \subseteq \tilde A_c$, where $A_c := \{X_{c,T_f}^{\pi^\star,\sup} \ge 2N^{(\infty)}-\beta\}$ and $\tilde A_c := \{X_{c,T_f}^{\pi^\star,\inf} \le \beta\}$. Hence,
$\mathbb{P}((A_1 \cup A_2 \cup A_3 \cup A_4)\cap A_d) \le \mathbb{P}(A_c)+\mathbb{P}(\tilde A_c)$. Combining with $A_d'$, we obtain
$\mathbb{P}(A_1 \cup A_2 \cup A_3 \cup A_4)
\le \mathbb{P}(A_c)+\mathbb{P}(\tilde A_c)+\mathbb{P}(A_d')$. This completes the proof.
\section{Proof of Lemma~\ref{lem:difference_process_upper_bound}}\label{appendix:difference_process_upper_bound}
Let $\hat{\tau}_{\beta} := \inf\{t \ge 0 : X_d^{\pi^\star}(t)=\beta\}$ denote the first hitting time of level $\beta$, so that $\mathbb{P}(\sup_{0\le t\le T_f} X_d^{\pi^\star}(t)\ge \beta)=\mathbb{P}(\hat{\tau}_{\beta}\le T_f)$. For $0<\hat{\beta}<\beta$, reaching $\beta$ requires first hitting $\hat{\beta}$. Hence, $\mathbb{P}(\hat{\tau}_{\beta}\le T_f)$ can be bounded by counting excursions of $X_d^{\pi^\star}(t)$ from $0$ to $\hat{\beta}$ over $[0,T_f]$, followed by the probability of reaching $\beta$ from $\hat{\beta}$ in each such excursion. To formalize this, we define epochs associated with $X_d^{\pi^\star}(t)$.

\begin{figure}
     \centering
     \includegraphics[scale=0.27,trim={2.45cm 0.5cm 2cm 0.1cm},clip]{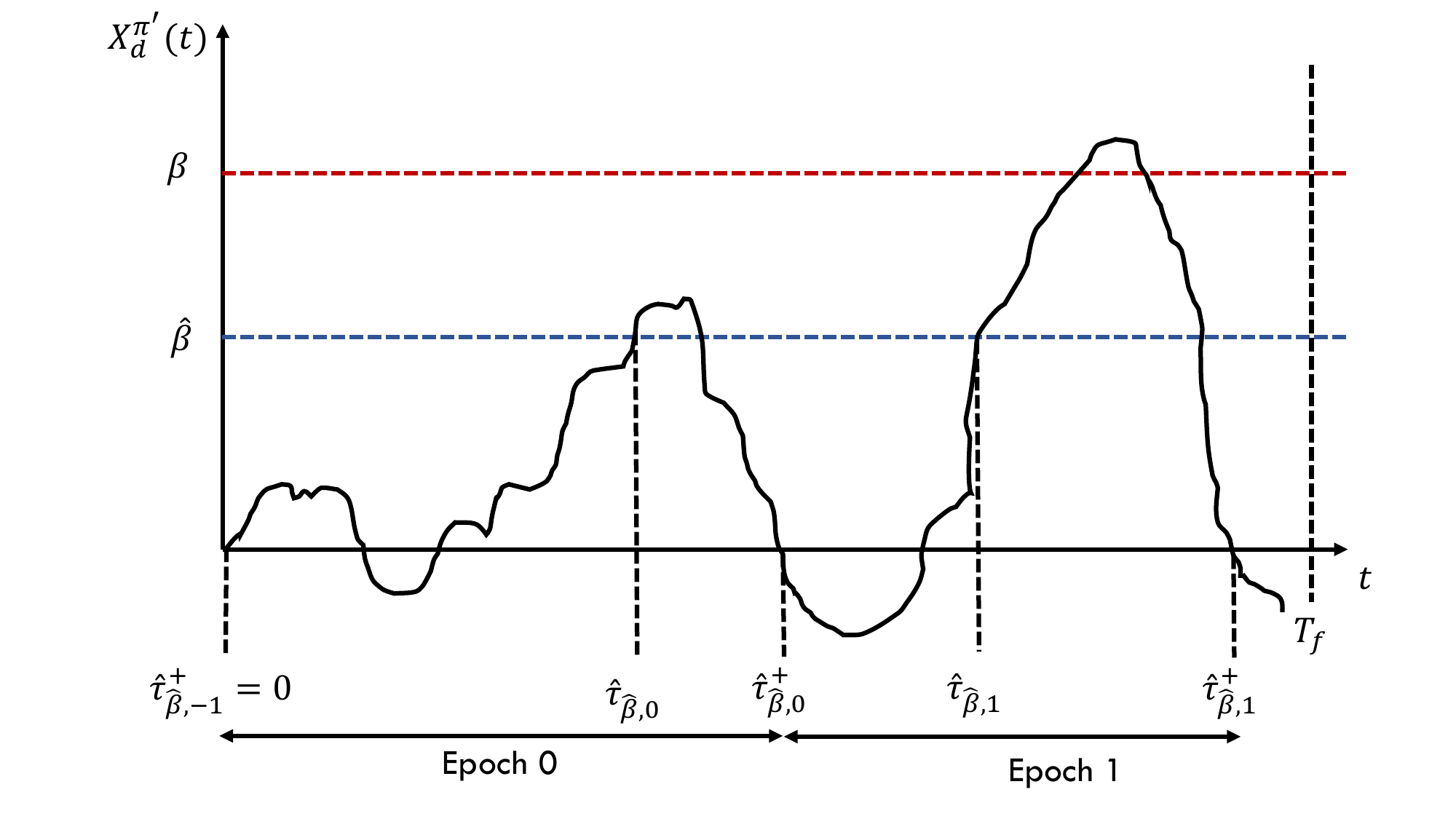}
     \caption{Illustration of epochs for a sample trajectory of $X_d^{\pi^\star}(t)$}
     \label{fig:epochs}
\end{figure}

Consider the process $X_d^{\pi^\star}(t)$ starting at $X_d^{\pi^\star}(0)=0$. Let $\hat{\tau}_{\hat{\beta},0}$ be the first hitting time of level $\hat{\beta}$, and let $\hat{\tau}_{\hat{\beta},0}^{+}$ denote the first return time to $0$ after $\hat{\tau}_{\hat{\beta},0}$. The interval $[0,\hat{\tau}_{\hat{\beta},0}^{+}]$ defines epoch $0$. For $i\ge 1$, define recursively $\hat{\tau}_{\hat{\beta},-1}^{+}:=0$,
$\hat{\tau}_{\hat{\beta},i}:=\inf\{t\ge \hat{\tau}_{\hat{\beta},i-1}^{+}: X_d^{\pi^\star}(t)=\hat{\beta}\}$, and
$\hat{\tau}_{\hat{\beta},i}^{+}:=\inf\{t\ge \hat{\tau}_{\hat{\beta},i}: X_d^{\pi^\star}(t)=0\}$. The interval $[\hat{\tau}_{\hat{\beta},i-1}^{+},\hat{\tau}_{\hat{\beta},i}^{+}]$ defines epoch $i$. The epochs are illustrated in Fig.~\ref{fig:epochs}.

We upper bound $\mathbb{P}(\hat{\tau}_{\beta}\le T_f)$ by counting the number of times $X_d^{\pi^\star}(t)$ hits $\hat{\beta}$ from $0$ over $[0,T_f]$, and bounding the probability of reaching $\beta$ from $\hat{\beta}$ after each such hit. Fix $T>0$ and consider the event $\{\hat{\tau}_{\hat{\beta},i}-\hat{\tau}_{\hat{\beta},i-1}^{+}\le T\}$, i.e., $\hat{\beta}$ is reached within time $T$ after the start of epoch $i$. Lemma~\ref{lem:upper_bound_1} provides an upper bound on this probability. Let $K$ denote the number of epochs in $[0,T_f]$ and define $K_{\max}:=T_f/T$. Corollary~\ref{cor:num_epochs} bounds $\mathbb{P}(K>K_{\max})$, showing that the number of epochs is at most $K_{\max}$ with high probability for large $T_f$. We now state Lemma~\ref{lem:upper_bound_1} and Corollary~\ref{cor:num_epochs}.
\begin{lemma}\label{lem:upper_bound_1}
    Let the transport capacity, $\transportCapacity$, $T > 0$, and $\hat{\beta} > 0$ be given. For any epoch $i\in \mathbb{N}$, define $\hat{\tau}_{\hat{\beta},-1}^{+}:=0$, $\hat{\tau}_{\hat{\beta},i}:=\inf\{t\ge \hat{\tau}_{\hat{\beta},i-1}^{+}: X_d^{\pi^\star}(t)=\hat{\beta}\}$, and $\hat{\tau}_{\hat{\beta},i}^{+}:=\inf\{t\ge \hat{\tau}_{\hat{\beta},i}: X_d^{\pi^\star}(t)=0\}$. Then
    \begin{align}
        \textstyle \mathbb{P}[\hat{\tau}_{\hat{\beta}, i} \leq \hat{\tau}_{\hat{\beta}, i-1}^{+} + T] \leq e^{-\big(\frac{\hat{\beta}^2}{2\sigma_d^2T}-\frac{2\B \hat{\beta}\transportCapacity}{\sigma_d^2}\big)}.
    \end{align}
\end{lemma}
\begin{proof}
    Please see Appendix~\ref{appendix:upper_bound_1}.
\end{proof}
\begin{corollary}\label{cor:num_epochs}
    Let the random variable $K$ denote the number of epochs occurring within the time interval $[0,T_f]$. Given $T > 0$, define $K_{max}:=\frac{T_f}{T}$. Then, for a given transport capacity $\transportCapacity$, and $\hat{\beta} > 0$, we have,
    \begin{align}
        \textstyle \mathbb{P}[K > K_{max}] \leq K_{max} e^{-\big(\frac{\hat{\beta}^2}{2\sigma_d^2T}-\frac{2\B \hat{\beta}\transportCapacity}{\sigma_d^2}\big)}.
    \end{align}
\end{corollary}
\begin{proof}
    Please see Appendix~\ref{appendix:num_epochs}.
\end{proof}

Having analyzed the number of epochs occurring in $[0,T_f]$, we now bound the probability of reaching $\beta$ from the intermediate level $\hat{\beta}$ within the horizon $T_f$. In particular, we consider the events $\{\hat{\tau}_{\beta,i} \le \hat{\tau}_{\hat{\beta},i}^{+},\, \hat{\tau}_{\hat{\beta},i}^{+} \le T_f\}$ and $\{\hat{\tau}_{\beta,i} \le T_f,\, \hat{\tau}_{\hat{\beta},i}^{+} > T_f\}$, where $\hat{\tau}_{\beta,i} := \inf\{t \ge \hat{\tau}_{\hat{\beta},i} : X_d^{\pi^\star}(t)=\beta\}$
denotes the first hitting time of $\beta$ in epoch $i$.
\begin{lemma}\label{lem:upper_bound_2}
    Let the transport capacity $\transportCapacity$ and levels $\beta > \hat{\beta} > 0$ be given. For each epoch $i \in \mathbb{N}$, define $\hat{\tau}_{\hat{\beta},-1}^{+}:=0$, $\hat{\tau}_{\hat{\beta},i}:=\inf\{t\ge \hat{\tau}_{\hat{\beta},i-1}^{+}: X_d^{\pi^\star}(t)=\hat{\beta}\}$, and $\hat{\tau}_{\hat{\beta},i}^{+}:=\inf\{t\ge \hat{\tau}_{\hat{\beta},i}: X_d^{\pi^\star}(t)=0\}$. Let $\hat{\tau}_{\beta,i}$ be defined as $\hat{\tau}_{\beta,i} := \inf\{t \ge \hat{\tau}_{\hat{\beta},i} : X_d^{\pi^\star}(t)=\beta\}$. Then, for any epoch $i \in \mathbb{N}$, we have
    \begin{align}
        &\textstyle \mathbb{P}[\hat{\tau}_{\beta, i} \leq \hat{\tau}_{\hat{\beta}, i}^{+}, \hat{\tau}_{\hat{\beta}, i}^{+} \leq T_f] \textstyle \leq e^{-\frac{2\B \transportCapacity(\beta-\hat{\beta})}{\sigma_d^2}}, \text{ and}\nonumber\\
        &\textstyle \mathbb{P}[\hat{\tau}_{\beta, i} \leq T_f, \hat{\tau}_{\hat{\beta}, i}^{+} > T_f] \textstyle \leq e^{-\frac{2\B \transportCapacity(\beta-\hat{\beta})}{\sigma_d^2}}.
    \end{align}
\end{lemma}
\begin{proof}
    Please see Appendix~\ref{appendix:upper_bound_2}.
\end{proof}
Finally, in Corollary~\ref{cor:num_epochs} and Lemma~\ref{lem:upper_bound_2}, we choose $\transportCapacity = c_p T_f^{-1/2+\epsilon}$, $T = \frac{1}{4\B c_p}T_f^{2\gamma-\epsilon/2}$, $K_{\max}=4\B c_p T_f^{1-2\gamma+\epsilon/2}$, $\hat{\beta}=T_f^{\gamma}$, and $\beta=T_f^{1/2-\epsilon/2}$, where $c_p>0$ and $0<\epsilon<\frac{1}{2}-\gamma$ and show that $\mathbb{P}(\hat{\tau}_{\beta}\le T_f)\to 0$ as $T_f\to\infty$. This follows since reaching $\beta$ from $\hat{\beta}$ within $[0,T_f]$ requires at least one successful epoch, while both the probability of exceeding $K_{\max}$ epochs (by Corollary~\ref{cor:num_epochs}) and the probability of reaching $\beta$ from $\hat{\beta}$ within an epoch (by Lemma~\ref{lem:upper_bound_2}) vanish as $T_f$ increases. We defer the detailed analysis to the following paragraphs.

Note that to reach level $\beta$, the process $X_d(t)$ must first hit $\hat{\beta}$ since $\beta>\hat{\beta}$ and $X_d(t)$ is continuous. Hence, $\{\hat{\tau}_\beta \le T_f\}$ occurs within some epoch. Let $K$ denote the number of epochs in $[0,T_f]$. Then $\mathbb{P}(\hat{\tau}_\beta \le T_f) \le \mathbb{P}(\cup_{i=0}^{K-1} G_i) + \mathbb{P}(\tilde{G}_K)$, where $G_i := \{\hat{\tau}_{\beta,i} \le \hat{\tau}_{\hat{\beta},i}^{+},\, \hat{\tau}_{\hat{\beta},i}^{+} \le T_f\}$ and $\tilde{G}_K := \{\hat{\tau}_{\beta,K} \le T_f,\, \hat{\tau}_{\hat{\beta},K}^{+} > T_f\}$. Let $K_{\max}$ be as defined earlier. Here, we can show that $\mathbb{P}(\cup_{i=0}^{K-1} G_i) \le \sum_{i=0}^{K_{\max}-1} \mathbb{P}(G_i) + \mathbb{P}(K>K_{\max})$. Applying Corollary~\ref{cor:num_epochs} and Lemma~\ref{lem:upper_bound_2} yields $\mathbb{P}(G_i) \le K_{\max}\Bigg[e^{-\frac{2\B U(\beta-\hat{\beta})}{\sigma_d^2}} + e^{-\left(\frac{\hat{\beta}^2}{2\sigma_d^2T}-\frac{2\B \hat{\beta}U}{\sigma_d^2}\right)}\Bigg]$ and $\mathbb{P}(\tilde{G}_K) \le e^{-\frac{2\B U(\beta-\hat{\beta})}{\sigma_d^2}}$. Thus we obtain $\mathbb{P}(\hat{\tau}_{\beta} \le T_f) \le K_{\max} e^{-\left(\frac{\hat{\beta}^2}{2\sigma_d^2T}-\frac{2\B\hat{\beta}U}{\sigma_d^2}\right)} + (K_{\max}+1)e^{-\frac{2\B U(\beta-\hat{\beta})}{\sigma_d^2}}$. With $U=c_p T_f^{-1/2+\epsilon}$, $\hat{\beta}=T_f^{\gamma}$, $\beta=T_f^{1/2-\epsilon/2}$, and $T=\frac{1}{4\B c_p}T_f^{2\gamma-\epsilon/2}$, this reduces to $\mathbb{P}(\hat{\tau}_{\beta} \le T_f) \le (K_{\max}+1)\exp\!\left[-\frac{2\B c_p}{\sigma_d^2}\left(T_f^{\epsilon/2}-T_f^{-1/2+\gamma+\epsilon}\right)\right]$. Since $K_{\max}=4\B c_p T_f^{1-2\gamma+\epsilon/2}$ grows polynomially while the exponential term decays faster than any polynomial for $0<\epsilon<\frac{1}{2}-\gamma$, it follows that $\lim_{T_f\to\infty}\mathbb{P}(\hat{\tau}_{\beta} \le T_f)=0$.
Thus,
\begin{align*}
    \textstyle \lim\limits_{T_f \to \infty} \mathbb{P}[\hat{\tau}_{\beta} \leq T_f] & \textstyle = \lim\limits_{T_f \to \infty} \mathbb{P}[\sup_{0 \leq t \leq T_f} X_d^{\pi^\star}(t) \geq \beta] = 0.
\end{align*}

A symmetric argument for $\mathbb{P}(\inf_{0\le t \le T_f} X_{d}^{\pi^\star}(t) \le -\beta)$ follows from Lemma~\ref{lem:upper_bound_1}, Corollary~\ref{cor:num_epochs}, and Lemma~\ref{lem:upper_bound_2}, yielding $\mathbb{P}(\inf_{0\le t \le T_f} X_{d}^{\pi^\star}(t) \le -\beta) \le K_{\max} e^{-\left(\frac{\hat{\beta}^2}{2\sigma_d^2T}-\frac{2\B \hat{\beta}\transportCapacity}{\sigma_d^2}\right)} + (K_{\max}+1)e^{-\frac{2\B \transportCapacity(\beta-\hat{\beta})}{\sigma_d^2}}$, and under $\transportCapacity=c_p T_f^{-1/2+\epsilon}$, $\hat{\beta}=T_f^{\gamma}$, $\beta=T_f^{1/2-\epsilon/2}$, and $T=\frac{1}{4\B c_p}T_f^{2\gamma-\epsilon/2}$, we obtain $\lim_{T_f\to\infty}\mathbb{P}(\inf_{0\le t \le T_f} X_{d}^{\pi^\star}(t) \le -\beta)=0$. This completes the proof.
\section{Proof of Lemma~\ref{lem:upper_bound_1}}\label{appendix:upper_bound_1}
\vspace{-0.2in}
\begin{align}\label{eq:middle_level_in_sup}
   \hspace{-0.2in} \mbox{Consider}, \quad  & \textstyle \mathbb{P}[\hat{\tau}_{\hat{\beta}, i} \leq \hat{\tau}_{\hat{\beta}, i-1}^{+}+T] \nonumber \\
    & \textstyle = \mathbb{P}[\sup_{\hat{\tau}_{\hat{\beta}, i-1}^{+} \leq t \leq \hat{\tau}_{\hat{\beta}, i-1}^{+}+T}X_d^{\pi^\star}(t) \geq T_f^{\gamma}]\nonumber\\
    & \textstyle = \mathbb{P}[\sup_{0 \leq h \leq T}X_d^{\pi^\star}(h+\hat{\tau}_{\hat{\beta}, i-1}^{+}) \geq T_f^{\gamma}]
\end{align}
Note that, for all $h \geq 0$, we have
\begin{align*}
    & \textstyle X_d^{\pi^\star}(h+\hat{\tau}_{\hat{\beta}, i-1}^{+}) = \sigma_d W_d(h+\hat{\tau}_{\hat{\beta}, i-1}^{+})\\
    & \textstyle -2\transportCapacity\int_0^{h+\hat{\tau}_{\hat{\beta}, i-1}^{+}} \sgn(X_d^{\pi^\star}(s))\text{d}s\\
    & \textstyle = X_d^{\pi^\star}(\hat{\tau}_{\hat{\beta}, i-1}^{+}) +\sigma_d[ W_d(h+\hat{\tau}_{\hat{\beta}, i-1}^{+}) - W_d(\hat{\tau}_{\hat{\beta}, i-1}^{+})]\\
    & \textstyle -2\transportCapacity\int_0^h \sgn(X_d^{\pi^\star}(v+\hat{\tau}_{\hat{\beta}, i-1}^{+}))\text{d}v.
\end{align*}

By the strong Markov property of Brownian motion, we have $\tilde{W}_d(h):=W_d(h+\hat{\tau}_{\hat{\beta}, i-1}^{+})-W_d(\hat{\tau}_{\hat{\beta}, i-1}^{+}), h \geq 0$, is again a Brownian motion [\cite{oksendal2013stochastic} Section 7.2, pp. 114]. Further by definition, $X_d^{\pi^\star}(\hat{\tau}_{\hat{\beta}, i-1}^{+})=0$. Therefore, $X_d^{\pi^\star}(h+\hat{\tau}_{\hat{\beta}, i-1}^{+}) = \sigma_d \tilde{W}_d(h) -2\transportCapacity\int_0^h \sgn(X_d^{\pi^\star}(v+\hat{\tau}_{\hat{\beta}, i-1}^{+}))\text{d}v$.
Note, for all $v\geq 0$ we have $\textstyle -2\transportCapacity\leq 2\transportCapacity\sgn(X_d^{\pi^\star}(v+\hat{\tau}_{\hat{\beta}, i-1}^{+})) \leq 2\transportCapacity$ and hence, for all $h\in[0,T]$,
\begin{align*}
    &  \hspace{-0.25in} \textstyle \sigma_d \tilde{W}_d(h) -2\B\transportCapacity h \leq X_d^{\pi^\star}(h+\hat{\tau}_{\hat{\beta}, i-1}^{+}) \leq \sigma_d \tilde{W}_d(h)+2\B\transportCapacity h.
\end{align*}
\begin{align}\label{eq:middle_level_inequality_1}
 \hspace{-0.5in} \mbox{Thus,} \  & \textstyle \mathbb{P}[\sup_{0\leq h \leq T} X_d^{\pi^\star}(h + \hat{\tau}_{\hat{\beta}, i-1}^{+}) \geq \hat{\beta}]\nonumber\\
    & \textstyle \leq \mathbb{P}[\sup_{0\leq h \leq T}(\sigma_d \tilde{W}_d(h)+2\B\transportCapacity h) \geq \hat{\beta}].
\end{align}
Let $\overline{\tau}_{\hat{\beta}} := \inf \{h \geq 0 : \sigma_d \tilde{W}_d(h)+2\B\transportCapacity h \geq \hat{\beta}\}$ denote the first time $\sigma_d \tilde{W}_d(h)+2\B\transportCapacity h$ reaches the level $\hat{\beta}$. Then, we have $\mathbb{P}[\sup_{0\leq h \leq T}(\sigma_d \tilde{W}_d(h)+2\B\transportCapacity h) \geq \hat{\beta}]=\mathbb{P}[\overline{\tau}_{\hat{\beta}} \leq T]$. Combining~\eqref{eq:middle_level_in_sup}, and~\eqref{eq:middle_level_inequality_1}, we have
\begin{align}\label{eq:middle_level_inequality_2}
    \textstyle \mathbb{P}[\hat{\tau}_{\hat{\beta}, i}\leq \hat{\tau}_{\hat{\beta}, i-1}^{+}+T] \leq \mathbb{P}[\overline{\tau}_{\gamma} \leq T].
\end{align}
Note that the process, $\sigma_d \tilde{W}_d(h)+2\B\transportCapacity h, h\in [0,T]$ is a Wiener process with positive drift of $2\B\transportCapacity$. Here using Girsanov's theorem, we can show that $\mathbb{P}[\overline{\tau}_{\gamma} \leq T] \leq e^{\frac{2\B\transportCapacity\hat{\beta}}{\sigma_d^2}-\frac{\hat{\beta}^2}{2\sigma_d^2T}}$. Thus from~\eqref{eq:middle_level_inequality_2}, we have $\mathbb{P}[\hat{\tau}_{\hat{\beta}, i} \leq \hat{\tau}_{\hat{\beta}, i-1}^{+} + T] \leq e^{-\big(\frac{\hat{\beta}^2}{2\sigma_d^2T}-\frac{2\B \hat{\beta}\transportCapacity}{\sigma_d^2}\big)}$.
This completes the proof.
\section{Proof of Corollary~\ref{cor:num_epochs}}\label{appendix:num_epochs}
Let the random variable $L_i := \hat{\tau}_{\hat{\beta}, i}^{+}-\hat{\tau}_{\hat{\beta}, i-1}^{+}$ denote the length of epoch $i\in \mathbb{N}$. Note that, $\hat{\tau}_{\hat{\beta}, i} \leq \hat{\tau}_{\hat{\beta}, i}^{+}$ by the definition of $\hat{\tau}_{\hat{\beta}, i}^{+}$. Therefore,
\begin{align*}
    &\textstyle \hat{\tau}_{\hat{\beta}, i}-\hat{\tau}_{\hat{\beta}, i-1}^{+} \leq \hat{\tau}_{\hat{\beta}, i}^{+}-\hat{\tau}_{\hat{\beta}, i-1}^{+}
    \implies \textstyle \hat{\tau}_{\hat{\beta}, i}-\hat{\tau}_{\hat{\beta}, i-1}^{+} \leq L_i.
\end{align*}
Thus for a given $T>0$, we have $L_i \leq T \implies \hat{\tau}_{\hat{\beta}, i}-\hat{\tau}_{\hat{\beta}, i-1}^{+} \leq T$, and hence $\mathbb{P}[L_i \leq T] \leq \mathbb{P}[\hat{\tau}_{\hat{\beta}, i}-\hat{\tau}_{\hat{\beta}, i-1}^{+} \leq T]$. Further, from Lemma~\ref{lem:upper_bound_1}, we have $\mathbb{P}[\hat{\tau}_{\hat{\beta}, i} \leq \hat{\tau}_{\hat{\beta}, i-1}^{+}+T] = \mathbb{P}[\hat{\tau}_{\hat{\beta}, i} - \hat{\tau}_{\hat{\beta}, i-1}^{+}\leq T] \leq e^{-\big(\frac{\hat{\beta}^2}{2\sigma_d^2T}-\frac{2\B \hat{\beta}\transportCapacity}{\sigma_d^2}\big)}$.
Therefore,
\begin{align}\label{eq:epoch_length_bound}
    \textstyle \mathbb{P}[L_i \leq T] \leq e^{-\big(\frac{\hat{\beta}^2}{2\sigma_d^2T}-\frac{2\B \hat{\beta}\transportCapacity}{\sigma_d^2}\big)}.
\end{align}
Note that, if the length of every epoch is greater than $T$, the number of epochs occurring within the time interval $[0, T_f]$ is less than $K_{max}=\frac{T_f}{T}$. Therefore,
\begin{align}\label{eq:num_epochs_to_epoch_length}
    \textstyle \mathbb{P}[K \leq K_{max}] \geq \mathbb{P}[\cap_{i=0}^{K_{max}-1} \{ L_i > T\}].
\end{align}
Further, $\mathbb{P}[\cap_{i=0}^{K_{max}-1} \{ L_i > T\}] = 1-\mathbb{P}[\cup_{i=0}^{K_{max}-1}\{ L_i \leq T\}] \geq 1 - K_{max}e^{-\big(\frac{\hat{\beta}^2}{2\sigma_d^2T}-\frac{2\B \hat{\beta}\transportCapacity}{\sigma_d^2}\big)}$,
where the last inequality follows from the union bound and~\eqref{eq:epoch_length_bound}. Hence, from~\eqref{eq:num_epochs_to_epoch_length},
\begin{align*}
    & \textstyle \mathbb{P}[K \leq K_{max}] \geq 1 - K_{max}e^{-\big(\frac{\hat{\beta}^2}{2\sigma_d^2T}-\frac{2\B \hat{\beta}\transportCapacity}{\sigma_d^2}\big)}\\
    \implies & \textstyle \mathbb{P}[K > K_{max}] \leq K_{max}e^{-\big(\frac{\hat{\beta}^2}{2\sigma_d^2T}-\frac{2\B \hat{\beta}\transportCapacity}{\sigma_d^2}\big)}.
\end{align*}
This completes the proof.
\section{Proof of Lemma~\ref{lem:upper_bound_2}}\label{appendix:upper_bound_2}
\vspace{-0.2in}
\begin{align}\label{eq:upper_level_joint_probability_1}
   \mbox{We have,} \ & \textstyle \mathbb{P}[\hat{\tau}_{\beta, i} \leq \hat{\tau}_{\hat{\beta}, i}^{+}, \hat{\tau}_{\hat{\beta}, i}^{+} \leq T_f]\nonumber \\
    &\hspace{-0.9in} \textstyle = \mathbb{P}[\hat{\tau}_{\beta, i} \leq \hat{\tau}_{\hat{\beta}, i}^{+} | \hat{\tau}_{\hat{\beta}, i}^{+} \leq T_f]\times\mathbb{P}[\hat{\tau}_{\hat{\beta}, i}^{+} \leq T_f]\nonumber\\
    &\hspace{-0.9in}  \textstyle =\mathbb{P}[\sup_{\hat{\tau}_{\hat{\beta}, i} \leq t \leq \hat{\tau}_{\hat{\beta}, i}^{+}}X_d^{\pi^\star}(t) \geq \beta | \hat{\tau}_{\hat{\beta}, i}^{+} \leq T_f]\times\mathbb{P}[\hat{\tau}_{\hat{\beta}, i}^{+} \leq T_f]\nonumber\\
    & \hspace{-0.9in}  \textstyle = \mathbb{P}[\sup_{0 \leq h \leq \hat{\tau}_{\hat{\beta}, i}^{+}-\hat{\tau}_{\hat{\beta}, i}}X_d^{\pi^\star}(h+\hat{\tau}_{\hat{\beta}, i}) \geq \beta | \hat{\tau}_{\hat{\beta}, i}^{+} \leq T_f]\times
    \nonumber\\
    & \hspace{-0.65in}  \textstyle \mathbb{P}[\hat{\tau}_{\hat{\beta}, i}^{+} \leq T_f].
\end{align}
Note that, for all $h \geq 0$, we have
$ X_d^{\pi^\star}(h+\hat{\tau}_{\hat{\beta}, i}) = \sigma_d W_d(h+\hat{\tau}_{\hat{\beta}, i})-2\transportCapacity\int_0^{h+\hat{\tau}_{\hat{\beta}, i}} \sgn(X_d^{\pi^\star}(s))\text{d}s = X_d^{\pi^\star}(\hat{\tau}_{\hat{\beta}, i}) +\sigma_d[ W_d(h+\hat{\tau}_{\hat{\beta}, i}) - W_d(\hat{\tau}_{\hat{\beta}, i})] -2\transportCapacity\int_0^h \sgn(X_d^{\pi^\star}(v+\hat{\tau}_{\hat{\beta}, i}))\text{d}v.$
By the strong Markov property of Brownian motion, we have $\tilde{W}_d(h):=W_d(h+\hat{\tau}_{\hat{\beta}, i})-W_d(\hat{\tau}_{\hat{\beta}, i}), h \geq 0$, is again a Brownian motion [\cite{oksendal2013stochastic} Section 7.2, pp. 114]. Further by definition, $X_d^{\pi^\star}(\hat{\tau}_{\hat{\beta}, i})=\hat{\beta}$. Therefore,
\begin{align}\label{eq:time_shifted_x_sde_1}
    & \hspace{-0.3in} \textstyle X_d^{\pi^\star}(h+\hat{\tau}_{\hat{\beta}, i}) = \hat{\beta}+\sigma_d \tilde{W}_d(h)\nonumber\\
    & \hspace{0.75in} \textstyle -2\transportCapacity\int_0^h \sgn(X_d^{\pi^\star}(v+\hat{\tau}_{\hat{\beta}, i}))\text{d}v.
\end{align}
By definition of $\hat{\tau}_{\hat{\beta}, i}^{+}$, we have, for all $t \in [\hat{\tau}_{\hat{\beta}, i}, \hat{\tau}_{\hat{\beta}, i}^{+}]$, $X_d^{\pi}(t) \geq 0$, or equivalently, for all $h \in [0, \hat{\tau}_{\hat{\beta}, i}^{+}-\hat{\tau}_{\hat{\beta}, i}]$, $X_d^{\pi^\star}(h+\hat{\tau}_{\hat{\beta}, i}) \geq 0$. Hence, $\sgn(X_d^{\pi^\star}(h+\hat{\tau}_{\hat{\beta}, i}))=1$ for all $h \in [0, \min(T_f, \hat{\tau}_{\hat{\beta}, i}^{+})-\hat{\tau}_{\hat{\beta}, i}]$. Thus, from~\eqref{eq:time_shifted_x_sde_1}, we have 
\begin{align}\label{eq:time_shifted_x_sde_2}
    & \textstyle X_d^{\pi^\star}(h+\hat{\tau}_{\hat{\beta}, i}) = \hat{\beta}+\sigma_d \tilde{W}_d(h)-2\transportCapacity h,
\end{align}
for all $h \in [0, \hat{\tau}_{\hat{\beta}, i}^{+}-\hat{\tau}_{\hat{\beta}, i}]$.

We define $Y_d(0):=\hat{\beta}$ and $Y_d(h):=Y_d(0)+\sigma_d\tilde{W}_d(h)-2\transportCapacity h$, and let $\tau_\beta^{Y}:=\inf\{h\ge0:Y_d(h)=\beta\}$ and $\tau^{Y+}:=\inf\{h\ge0:Y_d(h)=0\}$. From~\eqref{eq:time_shifted_x_sde_2}, we have $X_d^{\pi^\star}(h+\hat{\tau}_{\hat{\beta},i})=Y_d(h)$ a.s. for $h\in[0,\hat{\tau}_{\hat{\beta},i}^{+}-\hat{\tau}_{\hat{\beta},i}]$, implying $\tau^{Y+}=\hat{\tau}_{\hat{\beta},i}^{+}-\hat{\tau}_{\hat{\beta},i}$ a.s. Here we can show that,
$\mathbb{P}[\hat{\tau}_{\beta,i}\le \hat{\tau}_{\hat{\beta},i}^{+},\,\hat{\tau}_{\hat{\beta},i}^{+}\le T_f]
\le \mathbb{P}(\tau_\beta^{Y}\le T_f)
= \mathbb{P}(\sup_{0\le t\le T_f}Y_d(t)\ge\beta)$. We omit the detailed algebraic steps to conserve space. Since $Y_d$ is a Brownian motion with constant drift, using Girsanov’s theorem, we have
$\mathbb{P}(\tau_\beta^{Y}\le T_f)\le \exp\!\left(-\frac{2\B\transportCapacity(\beta-\hat{\beta})}{\sigma_d^2}\right)$. Similarly, we can show $\mathbb{P}[\hat{\tau}_{\beta,i}\le T_f,\,\hat{\tau}_{\hat{\beta},i}^{+}>T_f]
\le \mathbb{P}(\tau_\beta^{Y}\le T_f)
\le \exp\!\left(-\frac{2\B\transportCapacity(\beta-\hat{\beta})}{\sigma_d^2}\right)$. This completes the proof.
\section{Proof of Theorem~\ref{thm:upper_bound}}\label{appendix:upper_bound}
For $\resourceCapacity^{(\infty)}=c_{\infty}\sqrt{T_f}$ and $\beta=T_f^{(\frac{1}{2}-\frac{\epsilon}{2})}<\resourceCapacity^{(\infty)}$, Proposition~\ref{prop:conservative_sol_two_mg} gives
$\mathbb{P}[A_1 \cup A_2 \cup A_3 \cup A_4]
\leq \mathbb{P}[X_{c,T_f}^{\pi^\star,\sup} \geq 2\resourceCapacity^{(\infty)}-\beta]
+ \mathbb{P}[X_{c,T_f}^{\pi^\star,\inf} \leq \beta]
+ \mathbb{P}[\sup_{0\leq t\leq T_f}|X_d^{\pi^\star}(t)| \geq \beta]$.
Since $X_c^{\pi^\star}(t)=\resourceCapacity^{(\infty)}+\sigma_c W_c(t)$ with $\sigma_c=\sqrt{2}\sigma$, using hitting-time bounds for Brownian motion~\cite{karatzas2012brownian} yields
$\mathbb{P}[X_{c,T_f}^{\pi^\star,\sup} \geq 2\resourceCapacity^{(\infty)}-\beta]
\leq e^{-\frac{(\resourceCapacity^{(\infty)}-\beta)^2}{2\sigma_c^2 T_f}}$
and
$\mathbb{P}[X_{c,T_f}^{\pi^\star,\inf} \leq \beta]
\leq e^{-\frac{(\resourceCapacity^{(\infty)}-\beta)^2}{2\sigma_c^2 T_f}}$.
Thus,
$\mathbb{P}[A_1 \cup A_2 \cup A_3 \cup A_4]
\leq 2e^{-\frac{(\resourceCapacity^{(\infty)}-\beta)^2}{2\sigma_c^2 T_f}}
+ \mathbb{P}[\sup_{0\leq t\leq T_f}|X_d^{\pi^\star}(t)| \geq \beta]$.
With $\beta=T_f^{(\frac{1}{2}-\frac{\epsilon}{2})}$, using~\eqref{eq:difference_process_zero}, the second term vanishes as $T_f\to\infty$, while
\[
\lim_{T_f\to\infty} 2e^{-\frac{(\resourceCapacity^{(\infty)}-\beta)^2}{2\sigma_c^2 T_f}}=\delta,
\] completing the proof.
\bibliographystyle{elsarticle-num}
\bibliography{references}
\end{document}